\documentclass[11pt,a4paper]{article}
\usepackage{amsmath,amsthm,amssymb,indentfirst}
\usepackage[english]{babel} 
\usepackage[big]{layaureo}
\usepackage{xcolor}
\theoremstyle{definition}
\newtheorem{theorem}{Theorem}
\newtheorem*{lemma}{Lemma}
\newtheorem*{remark}{Remark}

\title{A Note on the Provision of a Public Service\\ of Different Qualities\thanks{We would like to thank Mauro Bambi, Daniel Cardona, and Leslie Reinhorn for their useful comments.}}
\author{Monica Anna Giovanniello\thanks{Departament d'Economia Aplicada, Universitat de les Illes Balears, 07122 Palma de Mallorca, Spain, ma.giovanniello@uib.cat}\thanks{Acknowledge financial support from the Spanish Ministerio de Ciencia y Educacion and Miniesterio de Universidades (Agencia Estatal de Investigación) through PID2019-107833GB-I00/AEI/10.13039/501100011033}\and Simone Tonin\thanks{Dipartimento di Scienze Economiche e Statistiche, Universit\`a degli Studi di Udine, 33100 Udine, Italy, simone.tonin@uniud.it}}
\date{November 2020}
\begin{document}
\maketitle
	
\begin{abstract}
We study how the quality dimension affects the social optimum in a model of spatial differentiation where two facilities provide a public service. If quality enters linearly in the individuals' utility function, a symmetric configuration, in which both facilities have the same quality and serve groups of individuals of the same size, does not maximize the social welfare. This is a surprising result as all individuals are symmetrically identical having the same quality valuation. We also show that a symmetric configuration of facilities may maximize the social welfare if the individuals' marginal utility of quality is decreasing.
\end{abstract}
	
\section{Introduction}
We study the problem of providing a public service in different locations. One may consider the problem of a government that has to provide education to its population and it has to decide both the locations of the schools and the level of resources (education quality) of each school. 
We build a model where individuals are uniformly distributed along the Hotelling line and the service is provided by a configuration of two facilities. We introduce a new dimension by assuming that each facility is not only characterized by its location but also by its quality. Individuals have homogeneous quality valuation and they bear quadratic transportation cost to consume the service. In this framework, we study the configuration of facilities that maximizes the social welfare function. A trade-off arises between the location and quality of the facilities and the size of the groups of individuals served by them.\par
	
We first consider the case in which quality enters linearly in the individuals' utility function. We show that, when the quality valuation is high enough, the optimal configuration is such that one facility has the highest quality and serves all the individuals. Differently, when the quality valuation is low enough, it is optimal to have a facility with the highest quality serving a large group of individuals and another one with the lowest quality serving a smaller group of individuals. It is surprising that the symmetric configuration, two facilities having the same quality and serving groups of individuals of the same size, is never optimal despite individuals have homogeneous quality valuation. 
	
We continue our analysis by considering an extension in which individuals' marginal utility of quality is decreasing. In this case, we show that the symmetric configuration of facilities is a social optimum when individuals' quality valuation is low enough. The intuition behind our results is the following. First, when individuals' marginal utility of quality is constant and the quality valuation is low enough, there is an individual who is indifferent between consuming the service in a distant high quality facility or in a closer one of low quality, \emph{i.e.} higher quality may compensate for higher transportation costs. Differently, when individuals' marginal utility of quality is decreasing and the quality valuation is low enough, the symmetric configuration of facilities becomes a social optimum as quality does not compensate for transportation costs.\par
	
To the best of our knowledge this is the first paper that analyse the provision of a public service where both location and quality are considered.\footnote{Neven and Thisse (1989) introduced vertically differentiation in an Hotelling model with private goods.} Starting from Tiebout (1956) there is a large literature on local public goods where their provision is constrained by a division of the space into jurisdictions. This literature mainly focuses on jurisdictions' optimal number, size, and composition (see Rubinfeld (1987) and Scotchmer (2002) for further references). 
Cremer et al. (1985) proposed an approach based on spatial competition theory to studies the optimal number and locations of facilities producing the public good.\footnote{See Gabszewicz and Thisse (1992) for a survey on spatial competition theory.} 
Differently from their contribution, we analyse the optimal provision of a public good characterized by both quality and location, but we astray from organizational and financing issues. \par
The rest of the paper is organized as follows. Section 2 describes the mathematical model. Section 3 is devoted to the study of the social optimum when individuals have constant marginal utility of quality. Section 4 considers a case in which individuals have decreasing marginal utility of quality. Section 5 concludes.
	
	\section{Mathematical model}
	We consider a spatial model described by a uniform distribution of individuals over the interval $I=[0,1]$ and by two facilities providing a public service without congestion.\par
	The two facilities are located in $a$ and $b$, with $0\leq a\leq b\leq 1$, and their quality are $q_a$ and $q_b$ respectively. The amount of resources available to finance the public service is exogenously given and it results in the following constraint $q_a+q_b=1$.\footnote{Given the individual's utility in (\ref{utility}), considering the case $q_a+q_b\leq 1$ would not change our analysis.}\par
	Individuals are assumed to have the same quality valuation for the public service but to differ for their locations. Each individual consumes the public service in only one location $x \in \{a, b\}$ and he has to bear the corresponding transportation cost. For simplicity, we also assume that all individuals consume the public service only once. We can then define the utility function of an individual located in $i\in I$ consuming the public service in $x\in\{a, b\}$ as
	\begin{equation}\label{utility}
	u(i,x)=\theta q_x - (i-x)^2,
	\end{equation}
	with $\theta>0$ representing the individuals' quality valuation for the public service.\footnote{We astray from considering the tax to finance the public service because subtracting a lump-sum tax from the utility of all individuals would not change our analysis.}\par 
	The social welfare function is then given by
	\begin{equation}\label{generalwelfare}
	\int_i \theta q_{x(i)}- \left(i-x(i)\right)^2 di,
	\end{equation}
	where $x(i)\in\mbox{arg}\max_{x\in\{a,b\}} \theta q_x-(i-x)^2$ is the optimal facility that maximizes the individual $i$'s utility function.
	
	\section{Social optimum}
	Let $(a,b,q)$, with $0\leq a\leq b\leq 1$ and $0\leq q\leq 1$, be a configuration of facilities that specifies their locations, $a$ and $b$, and their qualities, $q_a=q$ and $q_b=1-q$. A social optimum is a configuration of facilities that maximizes the social welfare function (\ref{generalwelfare}).\par
	To find the social optimum we exploit the fact that in our framework there exists a unique indifferent individual. Thus, we can write the social welfare function as a summation of two integrals: one defined over the interval of individuals consuming the public service in the facility located in $a$ and the other defined over the interval of individuals consuming the public service in the facility located in $b$. This allows us to use the first order necessary conditions for a maximum to find the social optimum.\par
	The function $\hat{j}(a,b,q)$ associates to any configuration of facilities the unique indifferent individual as follows
	\begin{equation*}\label{indifferent}
	\hat{j}(a,b,q)=\left\{ \begin{array}{ll}
	0 &\mbox{ if } (j<0) \mbox{ or } (a=b \mbox{ and } q<\frac{1}{2}),\\ 
	\frac{a^2-b^2-\theta(2 q-1)}{2(a-b)} &\mbox{ if } j\in I,\\
	1&\mbox{ if } (j>1) \mbox{ or } (a=b \mbox{ and } q>\frac{1}{2}),
	\end{array}\right.
	\end{equation*}
	with $j=\frac{a^2-b^2-\theta(2 q-1)}{2(a-b)}$ being the unique solution of the equation $u(i,a)=u(i,b)$ with $a \neq b$. When $j$ does not belong to the interval $I$, $\hat{j}(a,b,q)$ is equal to $0$ or $1$ depending of $j$ being less than 0 or greater than 1. Note also that $j$ is not defined at configurations of facilities where $a=b$ while, in such cases, $\hat{j}(a,b,q)$ is equal to $0$ or $1$ as the public service is consumed by all individuals in the facility with the highest quality. The only configuration of facilities for which $\hat{j}(a,b,q)$ is not defined is $\left(a,b,\frac{1}{2}\right)$ with $a=b$. However, as the following remark points out, this is not a problem for our analysis.
	\begin{remark} 
		Any configuration of facilities $\left(a,b,\frac{1}{2}\right)$ with $a=b$ is not a social optimum because it is always possible to find a configuration of facilities $\left(a',b',\frac{1}{2}\right)$ with $a'<b'$ that gives a higher social welfare. 
	\end{remark}
	Given the function $\hat{j}(a,b,q)$, the social welfare function can be written as
	\begin{equation}\label{welfare}
	W(a,b,q)=\int_{0}^{\hat{j}(a,b,q)}\theta q-(i-a)^2di+\int_{\hat{j}(a,b,q)}^{1}\theta (1-q)-(i-b)^2di,
	\end{equation}

	We now state our main result which identifies the social optima for different values of the quality valuation $\theta$.
	\begin{theorem}
		The social optima configurations of facilities are
		\begin{itemize}
			\item[--] for $\theta\in\left(0,\frac{1}{4}\right)$: $\left(\frac{1}{4}+\theta,\frac{3}{4}+\theta,1\right)$ and $\left(\frac{1}{4}-\theta,\frac{3}{4}-\theta,0\right)$;
			\item[--] for $\theta\in\left[\frac{1}{4},\infty\right)$: $\left(\frac{1}{2},b,1\right)$ with $b\in\left[\frac{1}{2},1\right]$ and $\left(a,\frac{1}{2},0\right)$ with $a\in\left[0,\frac{1}{2}\right]$.
		\end{itemize}
	\end{theorem}
	The key insight on which is based this result is that the marginal utility of quality is constant while the marginal transportation cost is increasing. Consequently, for $\theta\in\left[\frac{1}{4},\infty\right)$ at the social optimum there is one facility with quality 1 serving all individuals while the other one has quality $0$ and serves nobody. This configuration is optimal because for all individuals the gains derived from consuming in the facility with the highest quality more than compensate for the higher transportation costs borne to reach it. Note also that the individual who is located in the same point of the facility with quality $0$ prefers to consume from the highest quality facility. For $\theta\in\left(0,\frac{1}{4}\right)$ at the social optima there are again a facility with quality $1$ and a facility with quality $0$. But now individuals are segregated in two different groups. A large group of individuals consume the public service from the highest quality facility while the other smaller group is served from the facility with quality $0$. Heuristically speaking, for a small group of individuals quality does not compensate for the transportation costs and such group is then served by the nearer facility with quality $0$.
	
	\begin{proof}
		The social welfare maximization problem is
		\begin{equation}\label{maxwelfare}
		\begin{aligned}
		&\underset{a,b,q}{\text{max}}& &W(a,b,q),\\
		&\text{subject to}& & -a\leq 0 & (i)\\
		& & & a-b\leq 0 & (ii)\\
		& & & b\leq 1 & (iii)\\
		& & & -q\leq 0 & (iv)\\
		& & & q\leq 1 & (v)
		\end{aligned}
		\end{equation}
		To prove the theorem, we first find the configurations of facilities $(a_h^*,b_h^*,q_h^*)$ that satisfy the first order necessary conditions for a maximum and are candidate to be the social optimum. 
		We then compare the social welfare associated to each configuration of facilities found and we finally identify the social optimum.\par
		We first consider the first order necessary conditions for a maximum associated to triples $(a,b,q)$ such that $\hat{j}(a,b,q)\in(0,1)$. By solving the integrals of the social welfare function (\ref{welfare}), when $\hat{j}(a,b,q)\in(0,1)$, we obtain
		\begin{equation}\label{welfaresolved}
		\theta(1-q)-b^2+b-\frac{1}{3}-\frac{(a^2-b^2-\theta(2q-1))^2}{4\left(a-b\right)}.
		\end{equation}
		
		By the Lemma in the Appendix, when we consider function (\ref{welfaresolved}) to derive the first order necessary conditions for a maximum, we just focus on the cases in which only constraints $(iv)$ or $(v)$ may be binding. We then have three cases to analyze. First, neither constraints $(iv)$ nor $(v)$ are binding. Then, the first order necessary conditions for a maximum with respect to $a$, $b$, and $q$ are respectively
		\begin{align}
		&-\frac{\left(a^2-b^2-\theta (2q-1)\right) \left(3 a^2-4 a b+b^2+\theta(2 q-1) \right)}{4 (a-b)^2}=0,\label{partiala}\\
		&-\frac{((a-2) a-(b-2) b-\theta (2q-1)) \left(a^2+a (2-4 b)+b (3 b-2)-\theta (2q-1)\right)}{4 (a-b)^2}=0,\label{partialb}\\
		& \theta \left(\frac{\theta(1-2q)}{a-b}+a+b-1\right)=0\notag.
		\end{align}
		The unique solution of the system of equations above is the configuration of facilities
		$$(a_1^*,b_1^*,q_1^*)=\left(\frac{1}{4},\frac{3}{4},\frac{1}{2}\right).$$
		Second, only constraint $(iv)$ is binding. Then, the first order necessary conditions for a maximum with respect to $a$ and $b$ are still (\ref{partiala}) and (\ref{partialb}). While the first order necessary condition for a maximum with respect to $q$ becomes
		$$ \theta \left(\frac{\theta (1-2q) }{a-b}+a+b-1\right)+\lambda_{iv}=0, $$
		with $\lambda_{iv}$ being the Lagrangian multiplier associated to the constraint $(iv)$. Then, the configuration of facilities
		$$(a_2^*,b_2^*,q_2^*)=\left(\frac{1}{4}-\theta,\frac{3}{4}-\theta,0\right)$$
		for $\theta \in (0,\frac{1}{4})$ is the unique solution of first order necessary conditions for a maximum. 
		Third, only constraint $(v)$ is binding. By following, \emph{mutatis mutandis}, the same steps above, it is possible to verify that the configuration of facilities
		$$(a_3^*,b_3^*,q_3^*)=\left(\frac{1}{4}+\theta,\frac{3}{4}+\theta,1\right)$$ 
		for $\theta \in (0,\frac{1}{4})$ is the unique solution of first order necessary conditions for a maximum. \par
		
		We next consider the first order necessary conditions for a maximum associated to triples $(a,b,q)$ such that $\hat{j}(a,b,q)=0$. By solving the integrals of the social welfare function (\ref{welfare}), when $\hat{j}(a,b,q)=0$, we obtain
		$$\theta(1-q)-b^2+b-\frac{1}{3}.$$
		It is straightforward to verify that the configurations of facilities
		$$(a_4^*,b_4^*,q_4^*)=\left(a,\frac{1}{2},0\right),$$
		such that $a\in[0,\frac{1}{2}]$, are solutions of the first order necessary conditions for a maximum.\par
		
		We finally consider the first order necessary conditions for a maximum associated to triples $(a,b,q)$ such that $\hat{j}(a,b,q)=1$. By following, \emph{mutatis mutandis}, the same steps above, it is straightforward to verify that the configurations of facilities 
		$$(a_5^*,b_5^*,q_5^*)=\left(\frac{1}{2},b,1\right),$$
		such that $b\in[\frac{1}{2},1]$, are solutions of the first order necessary conditions for a maximum.
		\par
		
		At last, we compare the social welfare associated to each configuration of facilities that solves the first order necessary conditions for a maximum. It is immediate to calculate that
		\begin{align*}
		W(a_1^*,b_1^*,q_1^*)&=\frac{24 \theta-1}{48},\\
		W(a_2^*,b_2^*,q_2^*)=W(a_3^*,b_3^*,q_3^*)&=\frac{48 \theta^2+24 \theta-1}{48},\\
		W(a_4^*,b_4^*,q_4^*)=W(a_5^*,b_5^*,q_5^*)&=\frac{48 \theta-4}{48}.
		\end{align*}
		For $\theta\in(0,\frac{1}{4})$, the configurations of facilities $(a_h^*,b_h^*,q_h^*)$ for $h=1,2,3,4,5$ are candidate to be social optima. 
		It is straightforward to verify that 
		$$W(a_2^*,b_2^*,q_2^*)=W(a_3^*,b_3^*,q_3^*)>W(a_1^*,b_1^*,q_1^*)$$
		and that
		$$W(a_2^*,b_2^*,q_2^*)=W(a_3^*,b_3^*,q_3^*)>W(a_4^*,b_4^*,q_4^*)=W(a_5^*,b_5^*,q_5^*).$$
		We can now conclude that, for $\theta\in(0,\frac{1}{4})$, the social optima are $(a_2^*,b_2^*,q_2^*)$ and $(a_3^*,b_3^*,q_3^*)$. For $\theta\in[\frac{1}{4},\infty)$, the configurations of facilities $(a_h^*,b_h^*,q_h^*)$ for $h=1,4,5$ are candidate to be social optima.  Given that 
		$$W(a_4^*,b_4^*,q_4^*)=W(a_5^*,b_5^*,q_5^*)>W(a_1^*,b_1^*,q_1^*),$$
		the social optimal are $(a_4^*,b_4^*,q_4^*)$ and $(a_5^*,b_5^*,q_5^*)$, for $\theta\in[\frac{1}{4},\infty)$.
	\end{proof}


	\section{The symmetric configuration of facilities}
	
	In this section we study an example in which the symmetric configuration of facilities $\left(\frac{1}{4}, \frac{3}{4}, \frac{1}{2} \right)$ is a social optimum. The reason why it is worth investigating such configuration lies in its characteristics: all individuals consume the public service from a facility of the same quality, $q_a = q_b$, and both facilities serve groups of individuals of the same size. This case can then be considered as the most egalitarian provision of the public service.
	
	We now consider the same framework described in Section 2 with the exception that individuals have decreasing marginal utility of quality. Specifically, we assume the utility function of an individual located in $i \in I$ consuming the public service in $x \in \{a, b\}$ is
	\begin{equation}\label{utilityroot}
	\tilde{u}(i,x)=\theta \sqrt{q_x} - (i-x)^2.
	\end{equation}
	
	The function $\tilde{j}(a,b,q)$ that associates to any configuration of facilities the unique indifferent individual becomes
	\begin{equation*}\label{indifferent2}
	\tilde{j}(a,b,q)=\left\{ \begin{array}{ll}
	0 &\mbox{ if } (j<0) \mbox{ or } (a=b \mbox{ and } q<\frac{1}{2}),\\ 
	\frac{a^2-b^2-\theta(\sqrt{q}-\sqrt{1-q})}{2(a-b)} &\mbox{ if } j\in I,\\
	1&\mbox{ if } (j>1) \mbox{ or } (a=b \mbox{ and } q>\frac{1}{2}).
	\end{array}\right. 
	\end{equation*}
	abusing of notation let $j=\frac{a^2-b^2-\theta(\sqrt{q}-\sqrt{1-q})}{2(a-b)}$ be the unique solution of the equation $\tilde{u}(i,a)=\tilde{u}(i,b)$, with $a\neq b$. As above, the only configuration of facilities for which $\tilde{j}(a,b,q)$ is not defined is $\left(a,b,\frac{1}{2}\right)$ with $a=b$. However, the Remark holds also in this case.
	
	Given the function $\tilde{j}(a,b,q)$, the social welfare function can be written as
	\begin{equation}\label{welfareroot}
	\tilde{W}(a,b,q)=\int_{0}^{\tilde{j}(a,b,q)}\theta \sqrt{q}-(i-a)^2di+\int_{\tilde{j}(a,b,q)}^{1}\theta \sqrt{1-q}-(i-b)^2di.
	\end{equation}
	
	We now state our result on the symmetric configuration of facilities. 
	\begin{theorem}
		If $\theta \in (0,\frac{1}{4\sqrt{2}})$, then the social optimum is the symmetric configuration of facilities $\left(\frac{1}{4}, \frac{3}{4}, \frac{1}{2} \right)$ .
	\end{theorem}
	The key insight on which is based this result is that the marginal utility of quality is decreasing. In other words, for $\theta \in (0,\frac{1}{4\sqrt{2}})$, a high quality facility does not longer compensate for the higher transportation cost borne to reach it. Therefore, the symmetric configuration of facilities emerges as a social optimum.

	\begin{proof} The proof of Theorem 2 follows closely the one above. The social welfare maximization problem is still represented by (\ref{maxwelfare}) where, instead of having $W(a,b,q)$, we have $\tilde{W}(a,b,q)$ defined in (\ref{welfareroot}). We begin by considering the first order necessary conditions for a maximum associated to triples $(a,b,q)$ such that $\tilde{j}(a,b,q)\in(0,1)$.
		By solving the integrals of the social welfare function (\ref{welfareroot}), when $\tilde{j}(a,b,q)\in(0,1)$, we obtain 
		\begin{equation}\label{welfaresolvedroot}
		\theta\sqrt{1-q}-b^2+b-\frac{1}{3}-\frac{\left(a^2-b^2-\theta\left(\sqrt{q}-\sqrt{1-q}\right)\right)^2}{4 (a-b)}.
		\end{equation}
		The Lemma holds also when individuals have the utility function in (\ref{utilityroot}) and then at a social optimum $(a^*,b^*,q^*)$ with $\tilde{j}(a^*,b^*,q^*)\in (0,1)$ constraints $(i)$, $(ii)$, and $(iii)$ are not binding.
		Therefore, we focus on the cases in which only constraints $(iv)$ or $(v)$ may be binding. First, neither constraints $(iv)$ and $(v)$ are binding. Then, the first order necessary conditions for a maximum with respect to $a$, $b$, and $q$ are respectively
		
		\begin{align*}
		&-\frac{\left(a^2-b^2-\theta\left(\sqrt{q}-\sqrt{1-q}\right)\right) \left(3 a^2-4 a b+b^2+\theta\left(\sqrt{q}-\sqrt{1-q}\right)\right)}{4 (a-b)^2}=0,\\
		&-\frac{\left((a-2) a-(b-2) b-\theta\left(\sqrt{q}-\sqrt{1-q}\right)\right) \left(a^2+a (2-4 b)+b (3 b-2)-\theta\left(\sqrt{q}-\sqrt{1-q}\right)\right)}{4 (a-b)^2}=0,\\
		&\frac{1}{4}\theta\left(\frac{\left(\frac{1}{\sqrt{q}}+\frac{1}{\sqrt{1-q}}\right) \left(a^2-b^2-\theta\left(\sqrt{q}-\sqrt{1-q}\right)\right)}{a-b}-\frac{2}{\sqrt{1-q}}\right)=0\notag.
		\end{align*}
		For $\theta \in (0,\frac{1}{4 \sqrt{2}})$ there is a unique solution of the system of equation above is the configuration of facilities\footnote{We have solved these first order necessary conditions for a maximum by using a computer algebra system.}
		$$(a_1^*,b_1^*,q_1^*)=\left(\frac{1}{4},\frac{3}{4},\frac{1}{2}\right).$$
		In the case in which either constrain $(iv)$ or constrain $(v)$ is binding, there are no solutions of the first order necessary conditions for a maximum. \par
		
		We now consider the first order necessary conditions for a maximum associated to triples $(a,b,q)$ such that $\tilde{j}(a,b,q)=0$. By solving the integrals of the social welfare function (\ref{welfareroot}) we obtain
		$$\tilde{W}(a,b,q)=\theta\sqrt{1-q}-b^2+b-\frac{1}{3}.$$
		It is possible to verify that the configurations of facilities
		$$(a_2^*,a_2^*,q_2^*)=\left(a, \frac{1}{2}, 0 \right),$$
		such that $a\in[0,\frac{1}{2}]$, are solutions of the first order necessary conditions for a maximum.\par
		
		We finally consider the first order necessary conditions for a maximum associated to triples $(a,b,q)$ such that $\tilde{j}(a,b,q)=1$. By following, \emph{mutatis mutandis}, the same steps above, it is possible to show that the configurations of facilities
		$$(a_3^*,b_3^*,q_3^*)=\left(\frac{1}{2},b,1\right),$$
		such that $b\in[\frac{1}{2},1]$, are solutions of the first order necessary conditions for a maximum.\par
		
		At last, we calculate the social welfare level associated to each configuration of facilities that solves the first order necessary conditions for a maximum
		\begin{align*}
		\tilde{W}(a_1^*,b_1^*,q_1^*)&=\frac{24\sqrt{2}\theta-1}{48}\\
		\tilde{W}(a_2^*,b_2^*,q_2^*)=\tilde{W}(a_3^*,b_3^*,q_3^*)&=\frac{48\theta-4}{48}.\end{align*}
		Comparing the above levels of social welfare for each facility configuration, it is immediate to see that 
		$$\tilde{W}(a_1^*,b_1^*,q_1^*)>\tilde{W}(a_2^*,b_2^*,q_2^*)=\tilde{W}(a_3^*,b_3^*,q_3^*),$$
		for $\theta\in(0,\frac{1}{4 \sqrt{2}})$.
		But, then, the social optimum is $(a_1^*,b_1^*,q_1^*)$.
	\end{proof}
	
	\section{Conclusion}
	
	Our framework allows us to study the provision of a public service horizontally (location) and vertically (quality) differentiated. 
	Theorem 1 shows that the symmetric configuration of facilities is never a social optimum even if individuals have an identical and low enough quality valuation. This result differs from models where the public service is only horizontally differentiated in which, at the optimum, each facility serves a group of individuals of the same size for any number of facilities (see Cremer et al., 1985). Theorem 2 highlights the role of decreasing marginal utility of quality by showing that the symmetric configuration of facilities may arise as a social optimum when it is not possible to compensate higher transportation costs by providing a public service of higher quality. 
	Our results point to many possible extensions that are worthy of further study. It would be very interesting to endogenize the optimal number of facilities and to formalize the taxation scheme required to finance the provision of the public service.
	
	\appendix
	\section*{Appendix}
	To prove our main results we need the following lemma showing which constraints are not binding at the social optima with $\hat{j}(a^*,b^*,q^*) \in (0,1)$. 
	\begin{lemma}\label{interior}
		Let $(a^*,b^*,q^*)$ be a social optimum. If $\hat{j}(a^*,b^*,q^*) \in (0,1)$, then the constraints $(i)$, $(ii)$, and $(iii)$ are not binding.
	\end{lemma}
	\begin{proof}
		Let $(a^*,b^*,q^*)$ be the social optimum and assume that $\hat{j}(a^*,b^*,q^*)\in(0,1)$. Suppose that $(i)$ is binding. Consider a location $a'$ such that $a'>a^*=0$. First, note that $\hat{j}(a',b^*,q^*)>\hat{j}(a^*,b^*,q^*)$ as $\frac{\partial\hat{j}(a,b,q)}{\partial a}>0$. Then, it is straightforward to verify that 
		$$\int_{0}^{j(a',b^*,q^*)} u(i,a') di > \int_{0}^{j(a^*,b^*,q^*)} u(i,a^*) di + \int_{j(a^*,b^*,q^*)}^{j(a',b^*,q^*)} u(i,b^*) di.$$
		
		As the individuals $i \in [j(a',b^*,q^*), 1]$, who consume the public service in $b^*$, obtain the same utility under both configurations of facilities, it follows that $W(a',b^*,q^*) > W(a^*,b^*,q^*) $, a contradiction. Hence, constraints $(i)$ is not binding. By following, \emph{mutatis mutandis}, the same steps, it is easy to show that constraint $(iii)$ is not binding. Finally, constraint $(ii)$ is not binding as $\hat{j}(a^*,b^*,q^*)\in(0,1)$.
	\end{proof}
	
	\end{document}